\documentclass{SCIS2019}
\usepackage{setspace}
\usepackage{amsmath}
\usepackage{amsfonts}
\usepackage{graphicx}
\usepackage{amsthm}
\usepackage{amssymb}
\usepackage{mathrsfs}
\usepackage{stfloats}
\usepackage{color}
\usepackage{bm}
\usepackage{subfloat}
\usepackage{cite}
\usepackage{multirow}
\usepackage{makecell}

\newtheorem{observation}{Observation}
\begin{document}
\ArticleType{RESEARCH PAPER}
\Year{2022}
\Month{}
\Vol{}
\No{}
\DOI{}
\ArtNo{}
\ReceiveDate{}
\ReviseDate{}
\AcceptDate{}
\OnlineDate{}

\title{Reconfiguring wireless environment via intelligent surfaces for 6G: reflection, modulation, and security}
{Reconfiguring wireless environment via intelligent surfaces for 6G: reflection, modulation, and security}

\author[1]{Jindan XU}{}
\author[1]{Chau YUEN}{yuenchau@sutd.edu.sg}
\author[2]{Chongwen HUANG}{}
\author[3]{Naveed UL HASSAN}{}
\author[4,6]{\\George C. ALEXANDROPOULOS}{}
\author[5]{Marco DI RENZO}{}
\author[6]{M\'{e}rouane DEBBAH}{}

\AuthorMark{Author A}

\AuthorCitation{Xu J D, Yuen C, Huang C W, Hassan N U, Alexandropoulos G C, Renzo M D , Debbah M}


\address[1]{Singapore University of Technology and Design, Singapore {\rm 487372}, Singapore}
\address[2]{College of Information Science and Electronic Engineering, Zhejiang University, Hangzhou {\rm310007}, China}
\address[3]{Department of Electrical Engineering, Lahore University of Management Sciences, Lahore {\rm54792}, Pakistan}
\address[4]{Department of Informatics and Telecommunications, National and Kapodistrian University of Athens, Athens {\rm15784}, Greece}
\address[5]{Universit\'e Paris-Saclay, CNRS, CentraleSup\'elec,
Laboratoire des Signaux et Syst\`emes, Gif-sur-Yvette {\rm91192}, France}
\address[6]{Technology Innovation Institute, Masdar
City {\rm9639}, Abu Dhabi, United Arab Emirates}


\abstract{
Reconfigurable intelligent surface (RIS) has been recognized as an essential enabling technique for the sixth-generation (6G) mobile communication network.
Specifically, an RIS is comprised of a large number of small and low-cost reflecting elements whose parameters are dynamically adjustable with a programmable controller. 
Each of these elements can effectively reflect a phase-shifted version of the incident electromagnetic wave. 
By adjusting the wave phases in real time, the propagation environment of the reflected signals can be dynamically reconfigured to enhance communication reliability, boost transmission rate, expand cellular coverage, and strengthen communication security. 
In this paper, we provide an overview on RIS-assisted wireless communications.
Specifically, we elaborate on the state-of-the-art enabling techniques of RISs as well as their corresponding substantial benefits from the perspectives of RIS reflection and RIS modulation.
With these benefits, we envision the integration of RIS into emerging applications for 6G.
In addition, communication security is of unprecedented importance in the 6G network with ubiquitous wireless services in multifarious verticals and areas.
We highlight potential contributions of RIS to physical-layer security in terms of secrecy rate and secrecy outage probability, exemplified by a typical case study from both theoretical and numerical aspects.
Finally, we discuss challenges and opportunities on the deployment of RISs in practice to motivate future research.
}


\maketitle{}

\section{Introduction}
 

Owing to the high data traffic and massive connections for vertical applications, e.g., internet-of-everything (IoE), integrated sensing and communications (ISAC), and space-air-ground-ocean (SAGO) communications, in future wireless networks, the sixth-generation (6G) mobile communication network is expected to provide extremely superior performance compared to the fifth-generation (5G) network. Specifically, the key performance indicators (KPIs) of 6G require to support wireless communications with peak data rate in the order of Tbps, user-experienced data rate upto 1~Gbps, 100~GHz maximum bandwidth, etc \cite{6G_Zhang}. 
To meet these stringent requirements, it is a widely-recognized path toward 6G to further evolve the massive multiple-input multiple-output (MIMO) technology and deeply explore high-frequency spectrum.

By exploiting hundreds, or even thousands, of antennas to conduct beamforming, the hardware cost and power consumption of massive MIMO increase exponentially with the number of antennas \cite{MassMIMO_Channel, MassMIMO_Loading, MassMIMO_DAC}.
As a potential cost-effective and energy-efficient solution, reconfigurable intelligent surface (RIS) has recently attracted much attention and has been regarded as an evolutionary extension of massive MIMO \cite{RIS_Survey_ComSoc, RIS_Survey_Jian, RIS_Survey_Liang, RIS_JSAC, RIS_Overview_Pan, RIS_MARISA}.
Specifically, an RIS consists of a large number of small and low-cost passive reflecting elements whose parameters are dynamically reconfigurable with a programmable controller \cite{RIS_Mag}. Each of these elements can effectively reflect a phase-shifted version of the incident electromagnetic wave. By adjusting the wave phases in real time, the propagation environment is dynamically reconfigured to enhance communication efficiency \cite{RIS_EE_You, RIS_EE_Meta, RIS_Mag_smart, RIS_EuCNC_1, RIS_EuCNC_2, RIS_EE_SE, RIS_EE_Yuen, RIS_Phasenoise_EEVM, RIS_Disc_Hybrid}, expand cellular coverage \cite{Stat_CSI_Coverage_Yang, Stat_CSI_Coverage, Stat_BoardCoverage}, and strengthen communication security \cite{RIS_PLS_MultiRIS, RIS_PLS_SR_MIMO, RIS_PLS_SR_AN, RIS_PLS_SR_AN_MulUe, RIS_PLS_SOP, RIS_PLS_SOP_SR,  RIS_PLS_SOP_AR_PhaseNoise, RIS_PLS_SR_AN_EveRIS}. 
These potentials of RIS in physical layer also inspire feasibility design of medium access control (MAC) to enable multiple-device communications with the assistance of RISs \cite{RIS_MAC_AI, RIS_MAC}.
Unlike the existing technology of amplify-and-forward relay \cite{Relay_3, Relay_1, Relay_2}, RIS requires no radio frequency (RF) chains and performs as a passive array which does not have the ability of signal processing.
Since this passive array surface induces negligible power consumption, it is easy for RIS to be integrated into wireless environments, e.g., by being deployed into ceilings, building facades, etc.

On the other hand, the 6G KPI in terms of transmission rate ranging from Gbps upto Tbps requires large bandwidth, which is only available in the high-frequency spectrum, e.g., millimeter-wave (mm-Wave), terahertz (THz), and visible light (VL) frequency bands.
Compared to the popularly adopted sub-6 GHz spectrum, high-frequency electromagnetic radiation experiences high free-space loss and severe oxygen absorption.
Due to these effects, communications in these high-frequency band are dominated by highly directional line-of-sight (LoS) propagations \cite{mmWave_pathloss}.
Such communications are vulnerable to blockage by various obstacles, such as walls, furniture, humans, and vehicles, especially in complicated indoor and dense urban environments \cite{mmWave_jdxu}.
RIS is an efficient solution to deal with the vulnerability of high-frequency spectrum by providing additional reflected paths \cite{RIS_PLS_SR_mmWave_DAC, RIS_Thz_Holo_GaoFF, RIS_DRL_MultiHop_THZ, RIS_VLC, VLC, RIS_LiFi, RIS_Prop}.

In addition, communication security has been attracting vast investigations for the 6G network \cite{PLS_IoT, PLS_DAC_jdxu}. With ubiquitous wireless services in multifarious areas, communications are becoming more vulnerable to eavesdropping and we are going to face unprecedented security issues in the 6G network. In most current applications, the security is guaranteed by using encryption techniques which rely on the assumption of relatively limited computation ability of eavesdroppers. However, it has been initially motivated by the landmark discovery in \cite{PLS_Shannon} that strict secure communications are contrivable from the physical layer. Techniques of this kind, namely physical-layer security, are designed by exploiting physical-layer characteristics like measurements of channel propagations.
Particularly in \cite{PLS_mmWave_jdxu}, for the first time, the physical-layer security has been shown achievable under high-frequency mm-Wave channels by developing a practical secure transmission scheme in the beam domain. Exploring this high-frequency spectrum, RIS can natually be utilized to reconfigure the propagation environment so that the physical-layer security is enhanced by improving the communication reliability of legitimate terminals while suppressing the leakage of information to eavesdroppers. 

In this paper, we give an overview on RIS-assisted wireless communications for the 6G network.
Specifically, we focus on the two main functionalities of RISs, i.e., reflection and modulation, and elaborate on the state-of-the-art enabling techniques from these two perspectives. We expose that RIS provides a cost-and-power efficient solution to coverage extension and rate boosting for 6G communications.
Thanks to these benefits, we envision opportunities for integrating RISs with a wide variety of emerging applications in the 6G network, including high-frequency transmissions, ISAC, and SAGO communications.
Besides, we evince the paramount importance of security in the 6G network and highlight the significant contributions of RIS to guarantee secure communications at the physical layer.
In particular, RIS is capable of enhancing secrecy performances in terms of both metrics, i.e., secrecy rate and secrecy outage probability (SOP).
Finally, we point out a multitude of challenging open issues for further research.


The rest of this paper is organized as follows.
In Section~\ref{Sec_refl_modu}, we elaborate on the enabling techniques of RIS reflection and RIS modulation, and exposes their corresponding substantial benefits to wireless communications.
In Section~\ref{Sec_appl}, we discuss opportunities for integrating RIS with emerging 6G applications.
Due to the unprecedented importance of security in the 6G network, we elaborate in Section~\ref{Sec_PLS} on the contributions of RIS to physical-layer security, exemplified by a case study.
In Section~\ref{Sec_chal}, we point out several prominent challenges and opportunities for further research.
Conclusions are drawn in Section~\ref{Sec_conc}. 


\section{RIS reflection and modulation}
\label{Sec_refl_modu}

%

In order to reap the potential benefits of RIS, there are generally two ways of utilizing RIS for communication performance enhancement as illustrated in Fig.~\ref{Fig_RIS_Refl_Modu}. 
In particular, RIS can serve as a reflector to reflect the incident electromagnetic wave via superimposed passive beamforming.
Alternatively, RIS is also taken as a transmitter which carries additional information via phase modulation schemes. 
In this section, we discuss these essential enabling techniques as well as their corresponding substantial benefits from perspectives of RIS reflection and modulation.
          
\begin{figure}[tb]
\centering\includegraphics[width=1\textwidth]{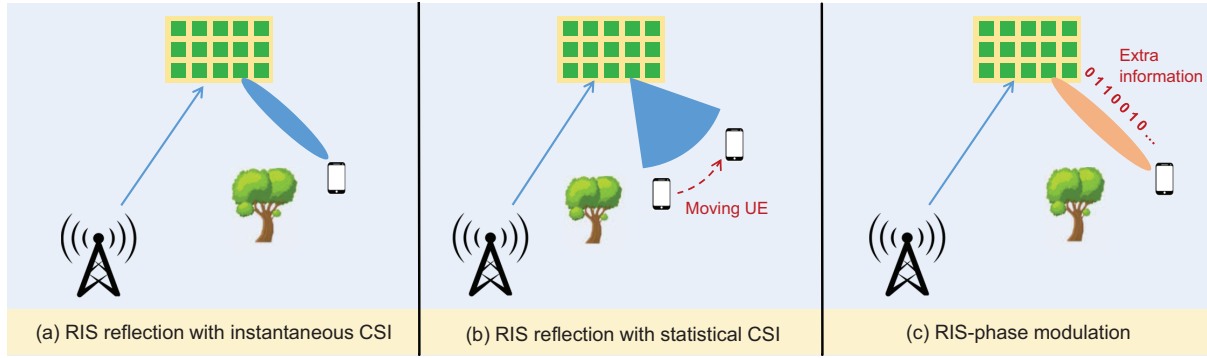}
\caption{RIS reflection and modulation.}
\label{Fig_RIS_Refl_Modu}
\end{figure}

\subsection{RIS reflection with instantaneous CSI}

As a typical reflection feature of RIS, energy efficient passive beamforming has been intensively investigated in existing works \cite{RIS_CE_SR, RIS_EE_Yuen, RIS_Phasenoise_EEVM, RIS_Disc_Hybrid, RIS_Robust_ImCSI_Lett, RIS_Robust_ImCSI, Stat_CSI_Han, Stat_CSI__MultiPair, Stat_CSI_Distributed, Stat_CSI_JBF_MISO, Stat_CSI_JBF, Stat_CSI_2Timescale, Stat_CSI_Coverage_Yang, Stat_CSI_Coverage, Stat_BoardCoverage, CoRIS_ZhaoYQ}.
It helps improve the communications in two ways.
On one hand, RIS reflection can establish virtual LoS paths to expand the cell coverage.
On the other hand, the received signal power can be enhanced via RIS beamforming.
However, the achievable performance realized by RIS-assisted communications, especially with massive MIMO, is usually limited by various hardware impairments, e.g., discrete phase shifts, imperfect channel state information (CSI), low-resolution analog-to-digital converters, and limited RF chains.
These constraints of hardware impairments are critical factors in the design of cost-and-power efficient RIS communication systems in practice.

One of these constraints of the most interest is that RIS is usually equipped with low-resolution phase shifters of a finite number of discrete adjustable phases.
In \cite{RIS_EE_Yuen}, discrete phase shifts for the RIS has been jointly designed with the transmit power for each user to maximize the energy efficiency in a multiuser multiple-input single-output (MISO) system. 
It is revealed that the RIS-assisted communications, even with coarse 1-bit phase shifters, significantly outperform the traditional relay-assisted communications in terms of energy efficiency.
The impacts of both phase errors at the RIS and imperfect RF chains at the base station (BS) have been firstly characterized in \cite{RIS_Phasenoise_EEVM}, which has discovered the essential theoretical observation that the spectral efficiency saturates due to the hardware impairments even when the numbers of elements at both the RIS and BS grow infinitely.
Alternatively, a hybrid beamforming method has been proposed in \cite{RIS_Disc_Hybrid}, considering discrete analog beamforming at the RIS and continuous digital beamforming at the BS for sum rate maximization.

Another popular constraint is that it is difficult to acquire perfect CSI especially in RIS-assisted communication systems. 
This is because RIS is a passive device without the ability of signal processing to transmit or receive RIS-specific pilots.
In order to tackle this challenge, a robust optimization of passive RIS beamforming with imperfect CSI has been developed in \cite{RIS_Robust_ImCSI_Lett}.
Then in \cite{RIS_Robust_ImCSI}, a general framework for the robust transmission design of an RIS-assisted multiuser MISO system has been established. 
Specifically, a transmit power minimization problem has been formulated subject to the worst-case rate constraints with imperfect CSI at the transmitter. 
In addition, a joint channel estimation and signal detection algorithm has been proposed in \cite{RIS_CE_SR} to simultaneously estimate the CSI and the transmitted signals.
    

\subsection{RIS reflection with statistical CSI}
As discussed in the above subsection, it is difficult to acquire perfect CSI in practice. Besides, the pilot and computation overhead to acquire instantaneous CSI of an RIS channel is usually unaffortable in some scenarios, e.g., with high mobility and when the reflecting surface is fabricated to be large. 
In addition, it appears infeasible to frequently adjust the phase shifts with highly dynamic channel fadings, due to limited response bandwidth of the diodes deployed on RIS.
To tackle these challenges, initial efforts have been devoted to studying coverage extension \cite{Stat_CSI_Coverage_Yang, Stat_CSI_Coverage, Stat_BoardCoverage}, reflection optimization\cite{Stat_CSI_Han, Stat_CSI__MultiPair, Stat_CSI_Distributed}, and joint beamforming design 
\cite{Stat_CSI_JBF_MISO, Stat_CSI_JBF, Stat_CSI_2Timescale} for RIS-assisted systems by utilizing statistical CSI as illustrated in Fig.~\ref{Fig_RIS_Refl_Modu}~(b).


Particularly for creating intelligent wireless environments, RIS has been employed to explicitly enhance cell-edge coverage against blockage.
By exploiting statistical CSI, a closed-form expression of the coverage area has been derived in \cite{Stat_CSI_Coverage_Yang} which quantitatively demonstrates that RIS is beneficial to extending the network coverage.
Utilizing the statistical CSI of large-scale channel fading, a closed-form expression of the coverage probability has been derived in \cite{Stat_CSI_Coverage}. 
It implies that the coverage probability scales up with $N^2$ where $N$ is the number of reflecting elements of the RIS.
Moreover, a quasi-static broad coverage has been designed in
\cite{Stat_BoardCoverage} by exploiting statistical CSI.
Considering the overhead of channel estimation, it has been verified that the proposed quasi-static broad coverage even outperforms the RIS reflection designs with instantaneous CSI.

For an RIS-assisted MISO communication system, the RIS reflection with statistical CSI has been optimized to maximize the ergodic spectral efficiency \cite{Stat_CSI_Han}. 
To compensate for the performance loss due to the lack of instantaneous CSI, a balanced choice of quantization resolution of phase shifters has been proposed, leading to a cost-effective system design.
In addition, an RIS-aided multi-pair communication system has been investigated in \cite{Stat_CSI__MultiPair}, under the availability of statistical CSI. An expression of the achievable rate has been derived and a genetic algorithm has been proposed for both continuous and discrete phase shift designs.
Utilizing the statistical channel correlation information and considering imperfect CSI, \cite{Stat_CSI_Distributed} has proposed a valuable RIS reflection method for distributed RISs which is helpful in promoting the use of RIS in practice. Moreover, the theoretical ergodic rate analysis in \cite{Stat_CSI_Distributed} also unveils that the distributed RIS deployment outperforms a conventional centralized architecture with the same number of total reflecting elements.

The absence of instantaneous CSI also brings challenges to the joint optimization of RIS beamforming and BS precoding.
Utilizing statistical CSI in Rician fading, \cite{Stat_CSI_JBF_MISO} has proposed an iterative algorithm to solve the joint optimization problem for an RIS-assisted MISO communication system.
In \cite{Stat_CSI_JBF}, a joint beamforming algorithm has been proposed for an RIS-assisted MIMO system by exploiting only the second-order momentum of the channel statistics. Specifically, the passive beamforming at the RIS and the transmit covariance matrix at the BS are alternatively optimized.
In \cite{Stat_CSI_2Timescale}, a two-timescale beamforming scheme has been presented, in which the passive beamforming at the RIS is first optimized based on the statistical CSI and the active precoding at the BS is then designed with instantaneous CSI.
RIS reflection design with statistical CSI provides an effective alternative to achieve better tradeoff between the overhead of channel estimation and performance for RIS-assisted communication systems especifically under highly dynamic wireless environments.

\subsection{RIS-phase modulation for rate boosting}
\label{RIS-phase modulation}
It has been intensively investigated that RIS reflection is effective in coverage extension and communication reliablity enhancement.
From the perspective of information theory, however, it has been recently discovered that the best way of fully reaping the advantage of RIS is to take it as not only a reflector but also a transmitter. 
In other word, extra information can be transmitted through RIS by, e.g., superimposing modulation on the adjustable phases of reflecting elements, as illustrated in Fig.~\ref{Fig_RIS_Refl_Modu}~(c), called superimposed  RIS-phase modulation \cite{RIS_Modu}.

It has been proved in \cite{RIS_Info} that joint encoding of RIS reflection and BS transmission can be the optimal capacity-achieving approach for an RIS-assisted communication system.
Then in \cite{RIS_PhaseModu_DoF}, it has been further revealed that the degree-of-freedom of RIS-assisted MIMO channel can be significantly improved by techniques like the RIS-phase modulation.
Based on an ON/OFF reflecting modulation, a framework of simultaneous passive beamforming and information transfer (PBIT) for multiuser MIMO communications has been investigated in \cite{RIS_PhaseModu_Onoff}.
Correspondingly, a sample average approximation based iterative algorithm has been developed for joint beamforming design at the transmitter, while a turbo message passing algorithm has been proposed to solve the bilinear estimation problem at the receiver.
Since the reflected signal power fluctuates over time with a various number of reflecting elements switched on, the outage probability of PBIT can be unacceptably high.
As a remedy, an RIS-based reflection pattern modulation (RIS-RPM) scheme has been proposed in \cite{RIS_PhaseModu_Patt}, in which a subset of reflecting elements are switched on to implement efficient beamforming at the cost of conveying less extra information via the RIS.
In these ON/OFF RIS-phase modulation schemes, the reflected power is less than maximum since only a part of reflecting elements are switched on.
To prevent from this excessive power loss, \cite{RIS_PhaseModu_Quad} has proposed a quadrature reflection modulation (RIS-QRM) method where all RIS elements remain active.
Specifically, the reflecting elements are partitioned into two subsets according to the local transmit data at the RIS. The phase shifts of the first subset are designed for RIS reflection, while the second subset are configured with orthogonal phase shifts, i.e., adding $\frac{\pi}{2}$. 

By playing the roles of both reflector and transmitter, RIS can distinctly enhance the transmission rate \cite{RIS_Info, RIS_PhaseModu_DoF}. However, most of the existing joint detectors are based on a maximum likelihood algorithm which is prohibitively complicated. The performance loss caused by other linear detectors can be unacceptable. Therefore, the design of RIS-phase modulation along with low-complexity detection methods is worthy of further study.






\section{Emerging RIS-empowered wireless applications for 6G}
\label{Sec_appl}

\begin{figure}[tb]
\centering\includegraphics[width=1\textwidth]{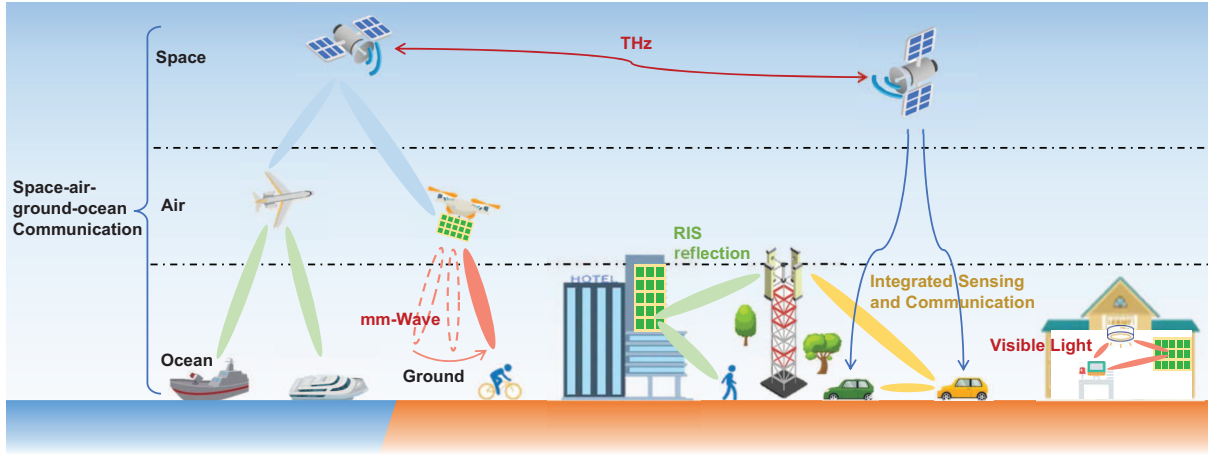}
\caption{Emerging RIS-empowered wireless applications for 6G.}
\label{Application}
\end{figure}

Acting as either a reflector or transmitter, RIS is promisingly beneficial to wireless communications via directional beamforming and phase modulation.
Thanks to these benefits, RISs have great potential in the 6G network for being integrated with emerging applications including high-frequency transmissions, ISAC, and SAGO communications, as illustrated in Fig.~\ref{Application}.

\subsection{RIS-assisted communications in high-frequency band}


Compared to 5G, the 6G network calls for order-of-magnitude increment in communication bandwidth in order to fulfill the KPI in terms of transmission rate. This is only possible  when exploring the high-frequency spectrum, e.g., mm-Wave, THz, and even VL frequency bands.
However, the path loss of high-frequency electromagnetic wave is commonly severe and the dominant LoS connectivity is easily blocked.
By reconstructing the wireless environment, RIS is potentially beneficial to enhancing the communication reliability in high-frequency spectrum.

In \cite{RIS_PLS_SR_mmWave_DAC}, an RIS has been utilized to mitigate the blockage challenge in mm-Wave secure communications. By applying a geometric model for sparse mm-Wave channels, the RIS reflecting and transmit beamforming have been jointly designed via an alternating optimization method. 
Compared to the mm-Wave radio spectrum ranging from 30~GHz to 300~GHz, the THz band provides more abundant bandwidth, i.e., from 0.1~THz to 10~THz.
However, transmission at the THz band suffers from extremely high free-space loss and is extraordinary sensitive to blockage.
For an RIS-assisted THz massive MIMO system, \cite{RIS_Thz_Holo_GaoFF} has proposed two beamforming approaches, i.e., narrow beam steering and spatial bandpass filtering, and a closed-loop channel estimation method.
Considering multiple RISs deployed to assist multiuser MISO communication, a joint design of analog beamformings at the RISs and digital beamforming at the BS has been proposed in \cite{RIS_DRL_MultiHop_THZ}.
In addition to the mm-Wave and THz frequencies, VL communication enjoys the advantage of abundant unlicensed optical spectrum \cite{VLC}.
In \cite{RIS_VLC}, VL communication has been utilized in an RIS-aided dual-hop indoor system.
To be specific, closed-form expressions of the outage probability and bit error rate (BER) have been analyzed and both these performance metrics are significantly improved by exploiting the RIS.
In addition, \cite{RIS_LiFi} has exposed that the integration of RIS with light fidelity technology enables plenty of innovative applications in the 6G network.
From these studies on high-frequency applications, it is concluded that RIS plays an important role in combating the large path loss and vulnerability due to blockage in high-frequency communications.

\subsection{RIS-empowered integrated sensing and communications}




ISAC has attracted great interest from both industry and academia for decades due to its manifold applications for object detection, tracking, positioning, etc \cite{ISAC_Holo, ISAC_ShiWei}. However, the performance of conventional ISAC is still constrained by unfavorable wireless environments, especially in the absence of LoS links. Recently, RIS has been incorporated into ISAC networks to improve the system performance via efficiently reconfiguring the communication environment \cite{RIS_ISAC_Location_Orientation, RIS_ISAC_Location, RIS_ISAC_Song, RIS_ISAC_Lee, RIS_BeamTrain, RIS_Spec_TVT, RIS_Spec_TWC, RIS_Spec_Magz}.

Sensing usually refers to estimate the positions and velocities of the targets via the reflected RF signals from the target terminals.
On one hand, reliable virtual LoS links constructed by RIS certainly improve the positioning accuracy of sensing.
An RIS-assisted positioning and communication architecture has been investigated in \cite{RIS_ISAC_Location_Orientation}.
Compared to a traditional system without RIS, the proposed scheme achieves up to two orders of magnitude reduction in positioning error.
Furthermore, \cite{RIS_ISAC_Location} has characterized the theoretical limits for an RIS-assisted ISAC system based on a general near/far-field signal model. It is revealed that the gain of positioning accuracy benefited from the RIS increases significantly with the carrier frequency.

On the other hand, RIS beamforming can improve the sensing signal-to-noise ratio (SNR) and expand the sensing range \cite{RIS_BeamTrain}.
To maximize the RIS's minimum beampattern gain towards a desired sensing direction, joint optimization of RIS beamforming and BS precoding has been studied in \cite{RIS_ISAC_Song}. 
Under the constrains of transmit power and communication SNR, the problem of sensing design is nonconvex such that an alternating optimization solution has been proposed to optimize the sensing beampattern while guaranteeing the communication SNR requirement.
Taking the sensing receiver design into account, an iterative algorithm has been proposed in \cite{RIS_ISAC_Lee} to improve the communication rate while guaranteeing a minimum sensing SNR.

In addition to the aforementioned sensing techniques, spectrum sensing also plays an important role in RIS-empowered ISAC systems.  
The incident signals towards an RIS are usually mixed with interfering signals, leading to a degraded received signal-to-interference-plus-noise ratio (SINR) at the desired receiver. 
To address this issue, a trained convolutional neural network (CNN) has been exploited at the RIS controller to infer the interfering signals from the incident signals \cite{RIS_Spec_TVT}.
By utilizing federated spectrum learning, the wireless bandwidth allocation, user-RIS association, and phase shift configuration have been jointly optimized in \cite{RIS_Spec_TWC}.
Then in \cite{RIS_Spec_Magz}, an RIS framework with intelligent spectrum sensing has been developed to exploit the inherent characteristics of the RF spectrum for a green 6G network.
All these studies evince enormous potential of RIS in assisting ISAC for 6G communications.

\subsection{Space-air-ground-ocean communications with RIS}
It has been a popular consensus that 6G will provide seamless global coverage, e.g., by the SAGO communication.
For SAGO applications with satellite, unmanned aerial vehicle (UAV), ships, and even underwater vehicles, the battery power of many devices is generally limited and it is inconvenient to frequently replace the built-in  batteries. RIS provides an effective solution for enabling energy efficient SAGO transmissions via passive beamforming. 
In \cite{RIS_Space}, an RIS-assisted SAGO architecture has been proposed where multiple RISs are integrated with mobile edge computing platforms to improve both communication and computing powers.
Besides, an RIS-assisted multiple access framework has been proposed in \cite{RIS_Access} which offers fairness for different types of terminals while reducing the computational complexity of RISs.

In addition, RIS can also extraordinarily enhance the diversity order when adopted to air communications \cite{UAV_Survey}.
In air-RIS systems, an RIS can be deployed on an airplane or a UAV to achieve additional flexibility \cite{RIS_Aeri}. 
As depicted in Fig.~\ref{Application}, this technique adds a dimension of mobility to RIS and improves its ability to establish LoS connectives, especially for terrains lacking infrastructure, e.g., in disaster scenes. 
Furthermore, extra degree of freedom can be provided by the rotation of RIS plane \cite{RIS_DoF}.

One last point coming with this application is that an entirely open environment of SAGO communications faces various unprecedented security risks, including physical, operational, network, and information threats, especially for underwater acoustic communications \cite{SAGIN_Security}.
RIS reflection helps direct a large amount of signal energy to desired terminals and areas, which in fact contributes to communication security.

These aforementioned works \cite{RIS_Space, RIS_Access, RIS_DoF, UAV_Survey, RIS_Aeri, SAGIN_Security} evince that SAGO communication can take the advantages of RIS to enhance energy efficiency, diversity order, and communication security. 

\section{RIS-empowered physical-layer security for 6G}
\label{Sec_PLS}
The development of cellular networks to 6G aims to serve all vertical applications and to connect things all around the world from ground to space. Communication security is therefore of unprecedented importance in the 6G wireless network. 
Different from conventional cryptographic methods implemented at the network and application layers \cite{cryptographic}, physical-layer security ensures secure transmissions from the fundamental perspective of information theory. 
It is completely impervious to the dramatically growing computational ability of advanced adversaries.
In this section, we elaborate on the contributions of RIS to physical-layer security from the perspectives of improving secrecy rate and reducing SOP, verified by a typical case study on RIS-assisted mm-Wave secure communications.  


\begin{table}[!t]
\centering
\footnotesize
\caption{RIS-empowered physical-layer security }
\label{table_PLS}
\tabcolsep 2pt 
\begin{tabular*}{\textwidth}{|c|c|m{5.5cm}|m{6.1cm}|}
 \cline{1-4}
  \textbf{Secure metrics} & \textbf{Reference} & \makecell[c]{\textbf{System setup}} & \makecell[c]{\textbf{Main contributions}} 
  \\\cline{1-4}
  \multirow{3}{*}{\makecell[c]{~\\~\\~\\~\\~\\~\\~\\~\\Secrecy rate}} & \cite{RIS_PLS_SR_MIMO} 
  & \makecell[l]{single RIS,\\ single user MISO, \\ single eavesdropper w. single antenna}
  & Secrecy rate maximization by joint optimization of RIS phase shifts and transmit signal covariance
  \\\cline{2-4} & \cite{RIS_PLS_MultiRIS} 
  & \makecell[l]{multiple RISs,\\ single user MISO, \\ single eavesdropper w. single antenna}
  & Secrecy rate maximization by joint optimization of RIS phase shifts and transmit precoding
  \\\cline{2-4}& \cite{RIS_PLS_SR_AN} 
  & \makecell[l]{single RIS, \\ single user MIMO, \\ single eavesdropper w. multiple antennas}  
  & {Secrecy rate maximization by joint optimization of RIS phase shifts, transmit signal covariance, and AN covariance}
  \\\cline{2-4} & \cite{RIS_PLS_SR_AN_MulUe} 
  & \makecell[l]{multiple RISs, \\ multiuser MISO, \\ multiple eavesdroppers w. multiple \\antennas} 
  & {Secrecy rate maximization by joint optimization of RIS phase shifts, transmit signal covariance, and AN covariance with imperfect CSI}
  \\\cline{2-4} & \cite{RIS_PLS_SR_AN_EveRIS} 
  & \makecell[l]{legitimate \& eavesdropping RISs, \\ single user MIMO, \\ single eavesdropper w. multiple \\antennas} 
  & {Secrecy rate maximization by joint optimization of RIS phase shifts, legitimate combining matrix, transmit precoding, and AN covariance}
  \\\cline{1-4}
  \multirow{3}{*}{\makecell[c]{~\\~\\~\\Secrecy outage\\probability}}  & \cite{RIS_PLS_SOP}
  & \makecell[l]{single RIS, \\ single user SISO, \\ single eavesdropper w. single antenna}
  & Asymptotic SOP analysis
  \\\cline{2-4} & \cite{RIS_PLS_SOP_SR} 
  & \makecell[l]{single RIS, \\ MIMO w. randomly located users, \\ single eavesdropper w. multiple antennas} 
  & SOP analysis in closed form and probability analysis of nonzero secrecy capacity
  \\\cline{2-4} & \cite{RIS_PLS_SOP_AR_PhaseNoise} 
  & \makecell[l]{single RIS, \\single user SISO, \\ multiple cooperative eavesdroppers} 
  & Exact SOP analysis under discrete phase shifters
  \\\cline{1-4}
\end{tabular*}
\end{table}

\subsection{RIS reflection design for enhanced secrecy rate}
Secrecy rate is a fundamental metric of physical-layer security to characterize the ability to prevent private information from eavesdropping \cite{RIS_PLS_SR_Key}.
RIS is capable of enhancing the secrecy rate from multifarious  aspects.
To be specific, RIS significantly improves the received signal power at legitimate terminals while suppresses the leakage of information to potential eavesdroppers via passive beamforming towards desired directions. 

The aim to maximize secrecy rate usually results in a nonconvex optimization problem due to the constraints of unit-modulus RIS phase shifts and the intricate objective of secrecy rate.
In order to find a locally optimal solution to this optimization problem with  coupled variables, alternating algorithms have been developed to optimize one variable with the others fixed \cite{RIS_PLS_MultiRIS, RIS_PLS_SR_MIMO, RIS_PLS_SR_AN, RIS_PLS_SR_AN_MulUe, RIS_PLS_SR_AN_EveRIS}.
Particularly in \cite{RIS_PLS_SR_MIMO}, an alternating algorithm has been proposed to jointly optimize the RIS phase shifts and transmit signal covariance for secrecy rate maximization in an RIS-assisted MIMO system in the presence of a passive eavesdropper equipped with a single antenna. This method is also extendable to a general scenario with a multi-antenna eavesdropper.
Then, multiple RISs have been considered in \cite{RIS_PLS_MultiRIS} where the RIS beamforming and BS precoding have been jointly optimized.
Further in \cite{RIS_PLS_SR_AN}, artificial noise (AN) has been adopted with the RIS to guarantee the secure transmission, where the covariance matrix of AN has also been jointly optimized with the RIS reflection. This additional optimization variable resulted in an extremely complicated nonconvex problem.
To address this issue, the coupled variables were alternately optimized while keeping the secrecy rate non-decreasing.
Moreover, practical concerns of imperfect CSI and scenarios with multiple eavesdroppers have been considered with RISs in \cite{RIS_PLS_SR_AN_MulUe}, where a robust alternating algorithm has been developed by utilizing the optimization techniques of successive convex approximation and semidefinite relaxation.
In addition, an eavesdropping RIS has been taken into account in \cite{RIS_PLS_SR_AN_EveRIS} to assist wiretap. It is revealed that positive secrecy rate is unachievable in the absence of a legitimate RIS, even for eavesdropping RISs equipped with small numbers of elements. 
Most of these schemes have reported significant performance gain in terms of secrecy rate provided by the optimized RIS reflection in a number of scenarios especially when no direct propagations exist.

\subsection{RIS-empowered communications with reduced secrecy outage}

In addition to secrecy rate, SOP, defined as the probability that the instantaneous secrecy rate falls below a target threshold, is an alternative popular metric to characterize the performance of physical-layer security \cite{RIS_PLS_SOP, RIS_PLS_SOP_SR, RIS_PLS_SOP_AR_PhaseNoise}.
Under the assumption of a large number of RIS elements, asymptotic SOP analysis has been conducted in \cite{RIS_PLS_SOP} for a basic RIS-assisted secrecy transmission model consisting of a transmitter, a legitimate user, and an eavesdropper.
In \cite{RIS_PLS_SOP_SR}, comprehensive secrecy performance analysis, including both SOP and average secrecy rate, has been presented by comparing the RIS-assisted communication to a conventional MIMO system without RIS.
Specifically, an exact expression of SOP has been derived under the assumption of randomly located users by applying stochastic geometry.
It is revealed that the secrecy performance in terms of SOP can be significantly improved with an increasing number of RIS elements.
Moreover, a closed-form expression of SOP with phase quantization errors has been derived in \cite{RIS_PLS_SOP_AR_PhaseNoise}, considering that it is infeasible to realize ideal continuous phase shifts at the RIS in practice.


In Table \ref{table_PLS}, we summarize the studies on RIS acting as a reflector to enhance secrecy rate and improve secrecy outage performance. As discussed in Section \ref{RIS-phase modulation}, RIS can also be utilized as a transmitter carrying extra information via phase modulation. This functionality also contributes to secure communications. Straightforwardly, more private information can be transmitted via RIS-phase modulation to desired user-equipments (UEs). On the other hand, AN can also be generated via RIS-phase modulation to help jamming the eavesdropper.
Therefore, it is still an interesting problem to investigate modulation design of RIS for secure communications in future wireless networks. 

\subsection{A case study on RIS-empowered communications with secure contour analysis}

In this section, we present a case study on secure mm-Wave communications to explicitly and quantitatively elaborate on the benefits of deploying an RIS for secure communications.
Specifically, we study a typical three-terminal wiretap channel where a multi-antenna transmitter (Alice) sends confidential signals to a desired receiver (Bob) in the presence of an eavesdropper (Eve).
Under the assumption that the LoS path between Alice and Bob/Eve is blocked, an RIS is exploited to assist the secure communication as displayed in Fig. \ref{PLS_Scenario}.


We assume that the RIS is equipped with $N$ reflecting elements while Alice exploits $M$ antennas. Bob and Eve are single-antenna receivers.
Due to the spatial sparsity of mm-Wave MIMO channels, assume that all the channels in this system are dominated by LoS paths. 
The Alice-RIS, RIS-Bob, and RIS-Eve channels are, respectively, denoted by $\mathbf{G}=l_A^{\frac{1}{2}}\mathbf{g}_R \mathbf{g}_A^H$, $\mathbf{h}_B=l_B^{\frac{1}{2}}\mathbf{g}_B$, and $\mathbf{h}_E=l_E^{\frac{1}{2}}\mathbf{g}_E$,
where $\mathbf{g}_A \in \mathbb{C}^{M\times 1}$, $\mathbf{g}_R \in \mathbb{C}^{N\times 1}$, $\mathbf{g}_B \in \mathbb{C}^{N\times 1}$, and $\mathbf{g}_E \in \mathbb{C}^{N\times 1}$ are signature response vectors
\footnote{
Specifically, $\mathbf{g}_A=\left[1,e^{-j2\pi\frac{d}{\lambda}\cos\phi},...,e^{-j2\pi(M-1)\frac{d}{\lambda}\cos\phi}\right]^T$ where $\phi\sim\mathcal{U}(0,\pi)$ is the angle of departure (AoD) of the transmit signal at Alice, $d$ denotes the distance between adjacent antennas, and $\lambda$ denotes the wavelength of the carrier frequency.
Notations $\mathbf{g}_R$, $\mathbf{g}_B$, and $\mathbf{g}_E$ are similarly defined with $\varphi\sim\mathcal{U}(0,\pi)$ denoting the angle of arrival (AoA) at the RIS, $\psi_B\sim\mathcal{U}(0,\pi)$ denoting the AoD of the reflected signals from RIS to Bob, and $\psi_E\sim\mathcal{U}(0,\pi)$ denoting the AoD of the reflected signals from RIS to Eve, respectively.
},
and $l_A$, $l_B$, and $l_E$ are large-scale fadings.

To transmit signal $s$ with a precoder vector $\mathbf{w}\in\mathbb{C}^{M\times 1}$ at Alice, assuming $\mathbb{E}\left\{|s|^2\right\}=1$ and $\|\mathbf{w}\|^2=1$, the received signal at Bob and Eve are respectively
\begin{align}
y_B= \mathbf{h}_B^H\mathbf{\Theta}\mathbf{G} \sqrt{P}\mathbf{w}s+n_B,
\label{y_B}
\end{align}
and
\begin{align}
y_E= \mathbf{h}_E^H\mathbf{\Theta}\mathbf{G} \sqrt{P}\mathbf{w}s+n_E,
\label{y_E}
\end{align}
where $n_B\sim \mathcal{CN}(0,\sigma_n^2)$ and $n_E\sim \mathcal{CN}(0,\sigma_n^2)$ denote the thermal noises at Bob and Eve, respectively, $\sigma_n^2$ is the noise power, and $\mathbf{\Theta}=\mathrm{diag} \{e^{j\theta_1},e^{j\theta_2},...,e^{j\theta_N}\}$ denotes the phase shift matrix at the RIS with $\theta_i$ being the phase shift of the $i$th reflecting element.

\begin{figure}[tb]
\centering\includegraphics[width=0.46\textwidth]{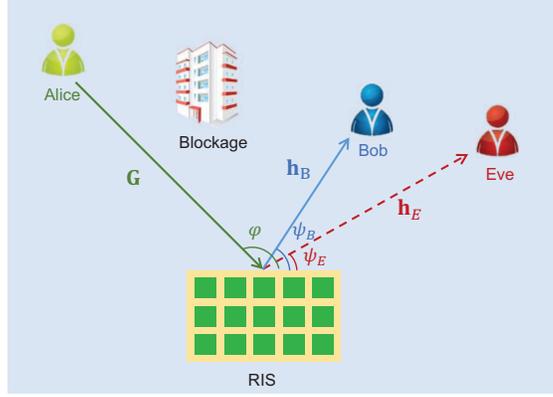}
\caption{A secure mm-Wave communication system assisted with an RIS.}
\label{PLS_Scenario}
\end{figure}

\subsubsection{ Theoretical secrecy rate analysis }

Consider the secrecy capacity defined in \cite{Sec_broadcast} as
\begin{align}
\label{Cs}
C_S=\max\limits_{s \rightarrow \mathbf{w}s \rightarrow y_B, y_E} I\Big(s;y_B\Big) - I\Big(s;y_E\Big),
\end{align}
where $I(\cdot;\cdot)$ denotes the mutual information between two random variables.
The secrecy capacity $C_S$ is characterized by maximizing over all joint distributions such that a Markov chain $s \rightarrow \mathbf{w}s \rightarrow \{y_B, y_E\}$ is formed.
Adopting a maximal-ratio-transmission (MRT) precoder at Alice, i.e., $\mathbf{w}=\frac{\left(\mathbf{h}_B^H\mathbf{\Theta}\mathbf{G}\right)^H}{\left\|\mathbf{h}_B^H\mathbf{\Theta}\mathbf{G}\right\|}$, the achievable secrecy rate is obtained as
\begin{align}
R_S&=\left[\log_2\left(1+\gamma_B\right)-\log_2\left(1+\gamma_E\right)\right]^+
,
\label{Rs}
\end{align}
where 
\begin{align}
\gamma_B
&=\frac{l_Al_BMP}{\sigma_n^2}\left|\mathbf{g}_B^H\mathbf{\Theta}\mathbf{g}_R \right|^2,
\label{SNR_B}
\end{align}
and
\begin{align}
\gamma_E
&=\frac{l_Al_E MP}{\sigma_n^2}\left| \mathbf{g}_E^H\mathbf{\Theta}\mathbf{g}_R\right|^2
,
\label{SNR_E}
\end{align}
are the received SNRs at Bob and Eve, respectively.

Comparing \eqref{SNR_B} with \eqref{SNR_E}, the differences between the desired SNR $\gamma_B$ and leakage SNR $\gamma_E$ come from two aspects. One is the large-scale fadings $l_B$ and $l_E$, which depend on the distances between the RIS to Bob and Eve.
The other is the beamforming gains $\left|\mathbf{g}_B^H\mathbf{\Theta}\mathbf{g}_R \right|$ and $\left| \mathbf{g}_E^H\mathbf{\Theta}\mathbf{g}_R\right|$, which relate to the cascade channels and RIS phase shifts.
To guarantee the secure communication, we analyze the secrecy performance with the optimal RIS phase shifts to improve $\gamma_B$ while suppress $\gamma_E$.


Without the knowledge of the presence of Eve, $\mathbf{\Theta}$ is optimized by maximizing the received SNR of Bob as in a conventional RIS communication system.
The optimal phase shifts at the RIS to maximize $\gamma_B$ in \eqref{SNR_B} is therefore
$\mathbf{\Theta}^*=\mathrm{diag}\{1,e^{j\theta},...,e^{j(N-1)\theta}\}$,
where
$\theta=2\pi\frac{d}{\lambda}(\cos\varphi-\cos\psi_B)$.
Substituting $\mathbf{\Theta}^*$ into \eqref{Rs}, the achievable secrecy rate is evaluated by
\begin{align}
R_S=\left[\log_2\left(1+\frac{l_Al_B MN^2P}{\sigma_n^2}\right)-\log_2\left(1+\frac{l_Al_E MP}{\sigma_n^2}\left| \mathbf{g}_E^H\mathbf{g}_B \right|^2\right)\right]^+
\label{Rs_opt}.
\end{align}

In \eqref{Rs_opt}, the first term is the achievable rate of Bob given the optimal RIS beamforming while the second term is the eavesdropping rate of Eve.
It implies that the transmit power can be scaled down by $\frac{1}{N^2}$ while maintaining a fixed desired rate at Bob.
For the leakage rate, we have $\left| \mathbf{g}_E^H\mathbf{g}_B \right|<N$ since Eve in general lies in a different direction from Bob. 
A positive secrecy rate is obtained to establish secure communications in this case. 
Two insightful observations are concluded in the following.

\begin{observation}
\label{Rs_Theorem_LowB}
Under the assumptions of $l_B\geq l_E$ and  $N\rightarrow\infty$, a tight lower bound for the ergodic achievable secrecy rate $R_S$ in \eqref{Rs_opt} is obtained as
\begin{align}
\underline{R}_S=\log_2\left(1+\frac{l_Al_B MN^2P}{\sigma_n^2}\right)-\log_2\left(1+\frac{l_Al_E MP}{\sigma_n^2}\eta\right),
\label{Rs_bound}
\end{align}
where we define
$\eta\triangleq N-\frac{2}{\pi^2}(N-1)+\frac{2N}{\pi^2}(\ln N +a)$
and $a$ is the Euler's constant.

\end{observation}

The proof of this observation is given in \ref{app_lemma_eta}.
From \eqref{Rs_bound}, the desired SNR at Bob grows in the order of $N^2$ while the leakage SNR at Eve increases in the order of $N\ln N$. 
Therefore, an increasing number of reflecting elements is beneficial to the ergodic secrecy rate. 
The quantitative scaling law is further characterized in the following observation.

\begin{observation}
For large $N\rightarrow\infty$ and $M\rightarrow\infty$, the lower bound for the ergodic achievable secrecy rate $\underline{R}_S$ in \eqref{Rs_bound} converges to
\begin{align}
\underline{R}_S\rightarrow\log_2\left(\frac{ N^2 l_B}{ \eta l_E}\right)
\rightarrow\log_2\left(\frac{ N \pi^2 l_B}{ 2 l_E}\right)
,
\label{Rs_bound_app}
\end{align}
where we apply $\log(1+x)\rightarrow\log(x)$ for large $x\rightarrow\infty$ and use the definition of $\eta$.
\end{observation}

This observation states that the ergodic achievable secrecy rate increases with $N$ while saturates for large $M$.
This is because $N$ reflecting elements at RIS provide a diversity gain in the order of $N^2$ for Bob while it is not identically beneficial to Eve.
In contrast, for increasing the number of antennas at Alice, the transmit beamforming gain is equally improved for both Bob and Eve.
It is therefore verified more efficient to increase the number of reflecting elements at RIS, which is also more energy efficient, than to increase the number of antennas at Alice.

\subsubsection{Numerical simulation}

For numerical simulations, we set the transmit power as $P=20$~W, the noise spectral density as $-174$~dBm/Hz, and the bandwidth as $100$~M. 
We exploit a standard linear model of the free space path loss of $l$ for mm-Wave communications as \cite{mmWave_pathloss}
\begin{align}
\label{pl}
-10\log_{10}(l)=\alpha+\beta10\log_{10}(d),
\end{align}
where $d$ is the distance in meters, and $\alpha$ and $\beta$ are respectively the least square fits of floating intercept and slope over the measured distances.
In particular, we use the values $\alpha=61.4$ and $\beta=2$ for $28$~GHz \cite{mmWave_pathloss}.


\begin{figure}[tb]
\centering
\subfloat[$N=8$, $d_A=10$, and $d_E=30$.]{
\begin{minipage}[t]{0.48\linewidth}
\centering
\includegraphics[width=1.1\linewidth]{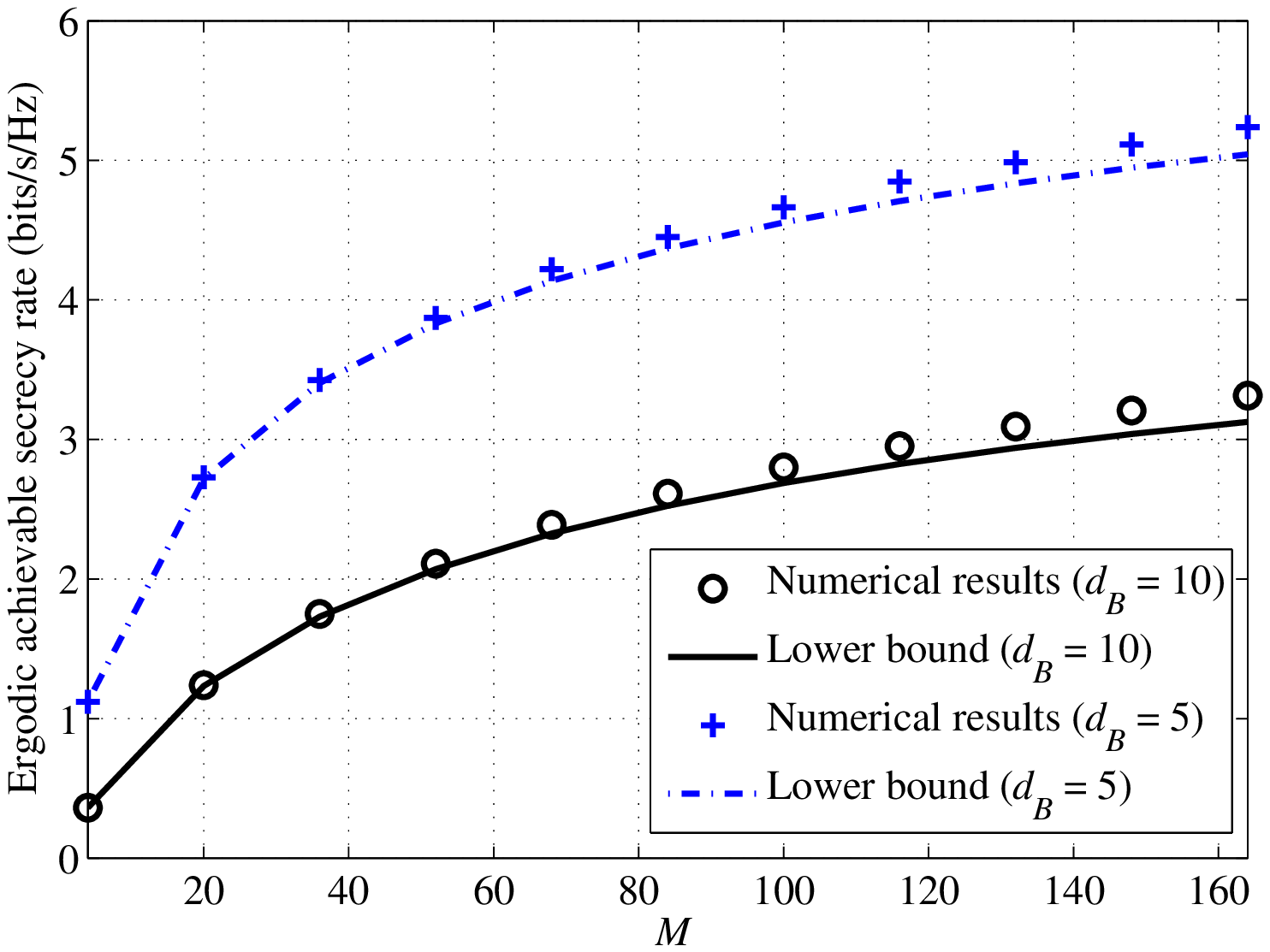}
\end{minipage}}
\subfloat[$d_A=15$, $d_B=20$, and $d_E=30$.]{
\begin{minipage}[t]{0.48\linewidth}
\centering
\includegraphics[width=1.1\linewidth]{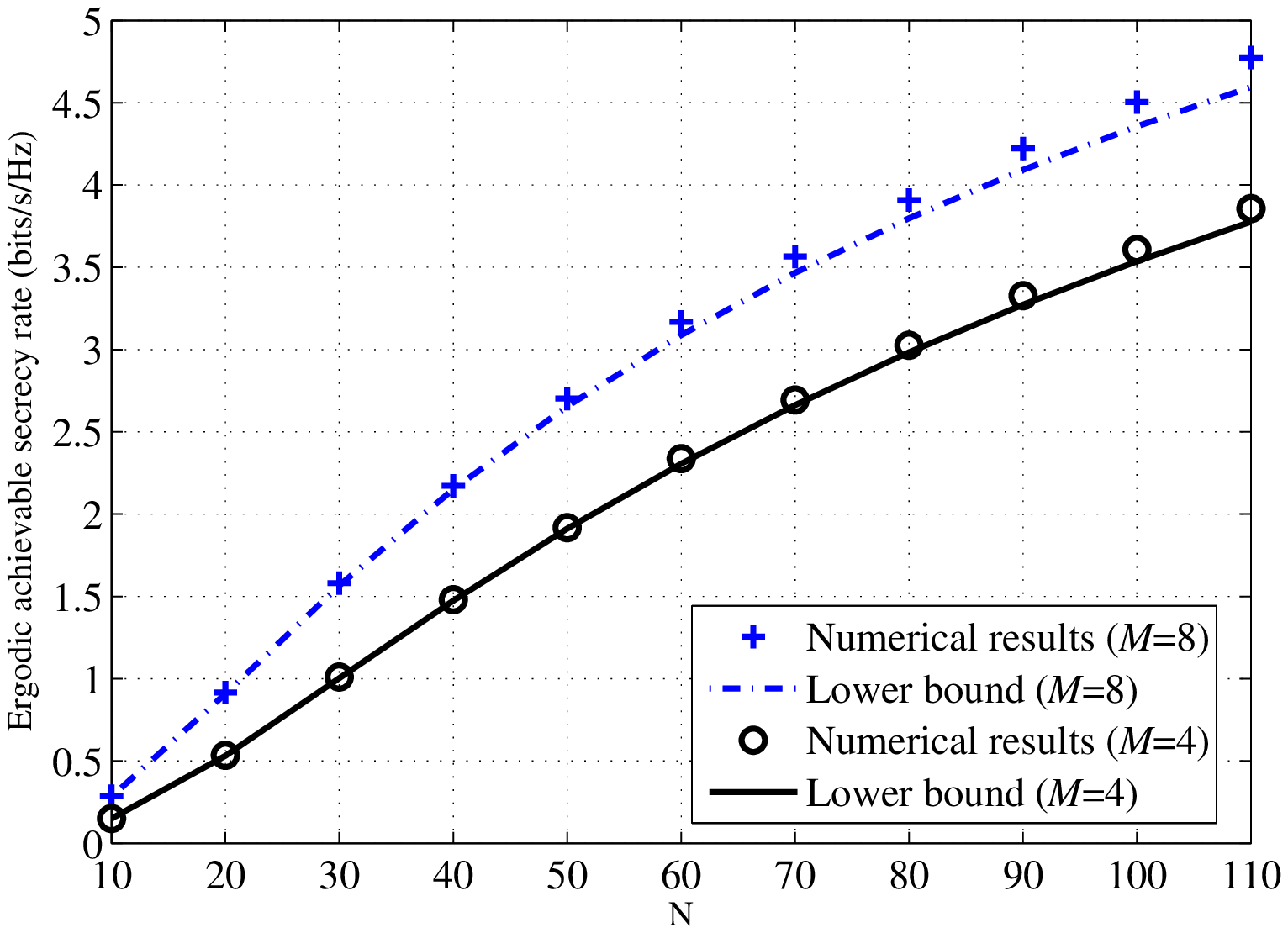}
\end{minipage}}
\caption{Ergodic achievable secrecy rate versus $M$ and $N$ ($\phi=\frac{\pi}{4}$ and $\varphi=\frac{3\pi}{4}$).}
\label{Ergodic_Rs}
\end{figure}

Fig.~\ref{Ergodic_Rs} verifies the tightness of the derived lower bound in \eqref{Rs_bound}. $l_A$, $l_B$, and $l_E$ respectively denotes the distances between the RIS and Alice, Bob, and Eve.
The secrecy rate increases with $M$ and finally tends saturated as shown in Fig.~\ref{Ergodic_Rs}(a).
A large number of transmit antennas provide more degrees of freedom at Alice, which is equally beneficial to both Bob and Eve. Its bonus to the secrecy rate is thus limited.
On the other hand, the secrecy rate monotonically increases with $N$ in Fig.~\ref{Ergodic_Rs}(b).
More efficient beamforming toward Bob can be realized by a larger number of reflecting elements at RIS, making it more challenging for Eve to wiretap.

Consider the particular scenario as illustrated in Fig.~\ref{PLS_Scenario}, where Bob and Eve are located with the reflected AoDs $\psi_B\in(0,\frac{\pi}{2})$ and $\psi_E\in(0,\frac{\pi}{2})$, respectively.
Given that Eve is fixed at a position, we define a secrecy area of Bob where the secrecy rate satisfies $R_S\geq R_0$ for a threshold secrecy rate $R_0\geq 0$.
The secrecy areas of Bob are illustrated in Fig.~\ref{area} given a fixed Eve marked by a pentagram.
The wheat, green, and black lines are contours of the ergodic secrecy rate corresponding to $R_0=1,2$, and $4$ bps/Hz, respectively.
In order to guarantee the target secrecy rate, e.g., $R_0=1$ bps/Hz illustrated by the wheat line in Fig.~\ref{area} (b), Bob is supposed to be located within 17.3 meters to the coordinate origin in the same direction of Eve, i.e., $\psi_B=\psi_E=\frac{\pi}{4}$.
This distance constraint can be relaxed to about 30 meters when Bob lies in different directions from Eve, i.e., $|\psi_B-\psi_E|>\frac{\pi}{12}$.
By comparing the subfigures in Fig.~\ref{area}, we observe that the secrecy area can be significantly expanded by increasing the number of RIS reflecting elements.

\begin{figure*}[t]
\centering
\subfloat[$M=1024$, $N=2$.]{
\begin{minipage}[t]{0.25\linewidth}
\centering
\includegraphics[width=1\linewidth]{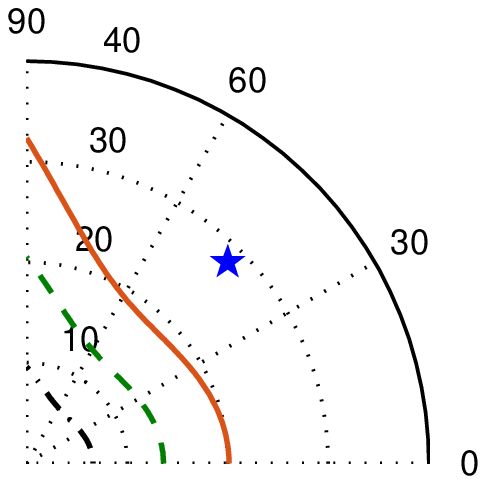}
\end{minipage}}
\subfloat[$M=32$, $N=8$.]{
\begin{minipage}[t]{0.25\linewidth}
\centering
\includegraphics[width=1\linewidth]{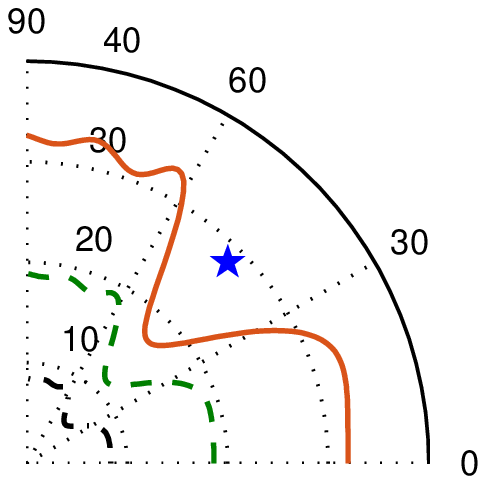}
\end{minipage}}
\subfloat[$M=32$, $N=16$.]{
\begin{minipage}[t]{0.25\linewidth}
\centering
\includegraphics[width=1\linewidth]{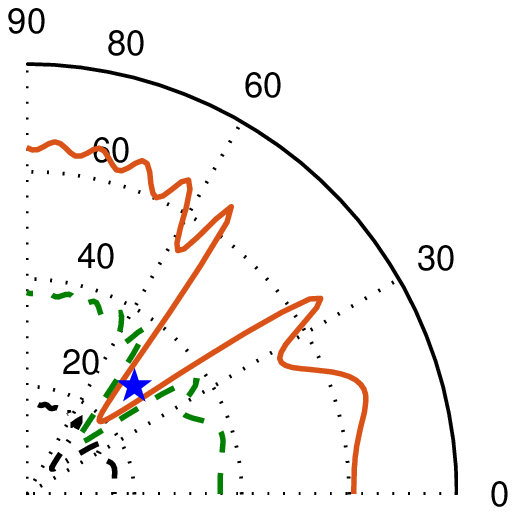}
\end{minipage}}
\subfloat[$M=16$, $N=32$.]{
\begin{minipage}[t]{0.25\linewidth}
\centering
\includegraphics[width=1\linewidth]{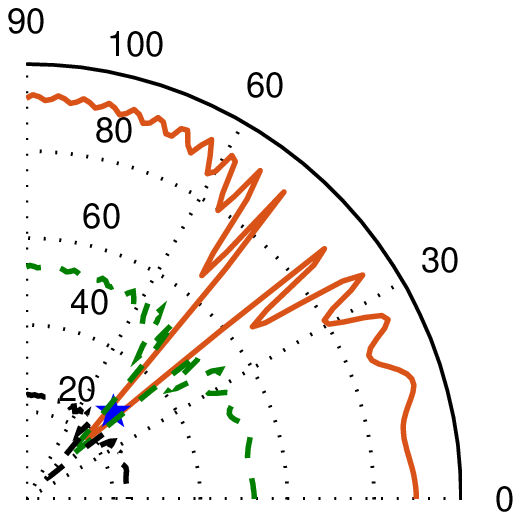}
\end{minipage}}
\caption{Secrecy area ($d_A=5\sqrt{2}$, $d_E=20\sqrt{2}$, $\psi_E=\frac{\pi}{4}$, $\phi=\frac{\pi}{4}$, and $\varphi=\frac{3\pi}{4}$).}
\label{area}
\end{figure*}

\section{Research challenges and opportunities}
\label{Sec_chal}
Emerging as a promising technology for 6G, the deployment of RIS in practice still faces a multitude of challenges and opportunities.
In this section, we discuss a few challenging issues and potential topics for future research.

\subsection{RIS channel estimation}

Channel estimation is an essential prerequisite in RIS-assisted communication systems. In the following, we overview the state-of-the-art estimation schemes for two different RIS structures, i.e., passive and active RISs.

\subsubsection{Channel estimation for passive RIS}
Typically, RIS is a passive device with no RF chains to actively transmit, receive, and process pilot sequences. It is therefore infeasible to conduct separate channel estimations for the BS-RIS and RIS-UE channels \cite{CE_SemiBlind, CE_OFDM}. 
In addition, the large size of reflecting array, resulting in a large number of channel coefficients to estimate, also makes the channel estimation extremely challenging. 

To estimate the cascaded BS-RIS-UE channel, a straightforward estimation framework is the ON/OFF method, in which the intricate estimation process is divided into multiple sub-phases and only a part of the RIS reflecting elements are switched on in each sub-phase.
An intuitive method is to switch on the $N$ reflecting elements one-by-one and estimate each cascaded channel coefficient in the one-by-one manner \cite{CE_Onoff}. Since an additional sub-phase is needed to estimate the direct BS-UE channel with all the elements off, totally $N+1$ sub-phases are required.
To reduce this huge overhead, this ON/OFF scheme has been modified in \cite{CE_Block} by adopting an element grouping method, where a sub-surface consisting of neighboring elements is turned on in each estimation sub-phase. 
Further in \cite{CE_DFT_MMSE}, the entire RIS is switched on in each sub-phase by applying discrete Fourier transform sequences.
Moreover, a three-phase channel estimation framework has been proposed in \cite{CE_3Phases_Correlation} by exploiting the correlations among the UE-RIS-BS cascade channels of different users.
On one hand, these ON/OFF estimation methods impose huge pilot overhead due to the multiple training sub-phases and prohibitive power consumption since the reflecting elements are switched on/off frequently. 
On the other hand, it is difficult to ensure the channel estimation accuracy because only a part of the RIS is switched on in each sub-phase while all RIS elements are used for data transmissions.

To achieve separate channel estimation, an alternative RIS channel estimation framework based on matrix decomposition has recently gained  attention.
In \cite{CE_Fact}, a bilinear generalized approximate message passing (BiG-AMP) algorithm has been adopted for sparse channel matrix decomposition.
Further in \cite{CE_Activity}, the channel estimation and the RIS-related activity detection problems have been jointly solved by the BiG-AMP, singular value thresholding, and vector AMP algorithms.
In addition, parallel factor decomposition has been considered as an efficient method for estimating RIS-assisted massive MIMO channels \cite{CE_Parafac, CE_Parafac_AMP}. This algorithm reduces the complexity of low-rank matrix estimations via decomposing a high dimensional matrix into a linear combination of multiple rank-one tensors.
Note that the efficiency of matrix decomposition relies heavily on the sparsity of the RIS-assisted MIMO channels which restricts its applications to general scenarios except for high-frequency communications.

\begin{table}[tb]
\footnotesize
\centering
\caption{RIS channel estimation}
\label{table_CE}
\tabcolsep 2pt 
\begin{tabular*}{\textwidth}{|c|c|c|m{4.6cm}|m{4.6cm}|}
 \cline{1-5}
  \textbf{\makecell[c]{RIS \\structure}} 
  & \textbf{\makecell[c]{Channel estimation\\ framework}} 
  & \textbf{Reference} & \makecell[c]{\textbf{Pros}} 
  & \makecell[c]{\textbf{Cons}} 
  \\\cline{1-5}
  \multirow{2}{*}{\makecell[c]{~\\Passive RIS}} & ON/OFF & \cite{CE_Onoff,CE_Block,CE_DFT_MMSE,CE_3Phases_Correlation} 
  & Low computational complexity
  & Excessive pilot overhead and prohibitive power consumption 
  \\\cline{2-5}
  & Matrix decomposition & \cite{CE_Fact,CE_Activity,CE_Parafac, CE_Parafac_AMP} 
  & Low pilot overhead
  & Only for sparse RIS channels
  \\\cline{1-5}
  \multirow{2}{*}{\makecell[c]{~\\Active RIS}} & Hybrid RIS & \cite{CE_active, CE_Hybrid_Sensing, CE_Hybrid_Sensing2}
  & Reduced overhead for channel estimation 
  & Additional hardware cost
  \\\cline{2-5}
  & Deep learning & \cite{RIS_CS_DL} 
  & Integrated RIS reflection design with channel estimation
  & Additional hardware cost
  \\\cline{1-5} 
\end{tabular*}
\end{table}

\subsubsection{Channel estimation for active RIS}

The passive RIS with no RF chains introduces negligible power consumption.
However, as aforementioned, it is difficult for this passive architecture to obtain accurate CSI.
Recently, low-power active sensors have been successfully installed at RIS, named active (or semi-passive) RIS, to trade off between the channel estimation performance and hardware/energy cost.

For an RIS with a smaller number of active sensors than that of the reflecting elements, two different active RIS architectures have been studied in \cite{CE_active, CE_Hybrid_Sensing, CE_Hybrid_Sensing2, RIS_CS_DL}.
One is similar to the structure of hybrid massive MIMO \cite{MassMIMO_Relay_Hybrid}, called hybrid RIS \cite{CE_Hybrid_Sensing}, in which a cluster of adjacent reflecting elements are connected to the same RF chain, acting as a sub-surface.
The phase shifts of all these reflecting elements in each sub-surface can be set either the same or different.
For this structure, the BS-RIS and RIS-UE channels can be separately obtained by sending training pilots to the RIS from the BS and UEs \cite{CE_active}.
Compared to purely passive RISs, the pilot overhead of this hybrid RIS is significantly reduced \cite{CE_Hybrid_Sensing2}. 
The other is to randomly choose a proportion of the RIS reflecting elements connected to the active sensors while letting the others remain passive \cite{RIS_CS_DL}.
It is revealed that the full BS-RIS and RIS-UE channels can be recovered from the sampled channel coefficients sensed at the active elements via compressive sensing.

Despite of these fruitful studies as summarized in Table~\ref{table_CE}, it is still an open problem to achieve accurate RIS channel estimation at acceptable cost especially under the consideration of practical hardware impairments.

\subsection{Interplay of RIS with artificial intelligence}

In order to serve multifarious applications for 6G by dynamically manipulating the wireless environments with RIS, it would be beneficial to introduce machine learning, especially deep learning (DL) technologies~\cite{RIS_AI_Renzo, RIS_Learning, Edge_Learning}. 
It is usually less tractable to solve the joint optimization of RIS reflection and BS precoding due to the varying wireless environment, moving terminals, and the sensitive interactions between active and passive beamformings, especially with extra nonconvex constraints and practical hardware impairments. It has been revealed that DL methods can help solve complicated nonconvex optimization problems with better solutions at lower cost \cite{RIS_AI}. Recently, intensive efforts have been devoted to studying the interplay of RIS with DL techniques \cite{RIS_DRL, RIS_DRL_MultiHop_THZ, RIS_DRL_EE, RIS_DL_CE, RIS_DL_Detect, RIS_DL_bandits, RIS_DL_Indoor, RIS_DL_1bit, RIS_DL_Uplink, RIS_DL_Scatter, RIS_Supervised} .

On one hand, a typical way of exploiting DL for RIS-assisted wireless communication is to treat the entire transmission system as a black-box deep neural network (DNN) \cite{RIS_DL_bandits, RIS_DL_Indoor, RIS_DL_Scatter}.
In such a data-driven DL framework, no expert knowledge of wireless environments or specific mathematical formulations is required.
By using reinforcement learning with DL, the target action of RIS can be learnt by observing instant rewards of rial-and-error interactions with the given state.
Specifically in \cite{RIS_DRL}, deep reinforcement learning (DRL) has been adopted to jointly design the RIS reflection and BS precoding for multiuser MISO communications.
This DRL-based technique has been extended to a multi-hop RIS-empowered system in \cite{RIS_DRL_MultiHop_THZ}.
To maximize the energy efficiency of a downlink multiuser MISO network, a DRL-based framework has been proposed in \cite{RIS_DRL_EE} to jointly design the phase shifts at the RIS and the transmit power at the BS.
Further to optimize the joint reflection of a plurality of RISs, a low-complexity supervised learning approach has been presented in \cite{RIS_Supervised}.
Exploiting 1-bit phase shifters at the RIS, the reflection coefficients can be considered as a binary vector. For this practical deployment, a deep Q-network has been utilized to explore the binary action spaces \cite{RIS_DL_1bit}.
In addition to the downlink transmissions studied in these aforementioned works, for an RIS-assisted multiuser uplink framework, a DL-based approach has been proposed in \cite{RIS_DL_Uplink} to jointly optimize the RIS reflection and symbol detection.
Different from existing alternating optimization algorithms, DL/DRL-based solutions are appropriate to solving joint optimization problems since it can achieve multiple target variables simultaneously as the outputs of a DNN.

On the other hand, a model-driven DL structure has recently attracted increasing attention.
It is of great interest to construct DL networks which are inspired by mathematical models of RIS communication module designs, including channel estimation, data detection, beam management, etc.
Utilizing a portion of active reflecting elements equipped for channel acquisition, a DL-based channel estimation method has been proposed in \cite{RIS_DL_CE}.
The cascade channel through active elements are obtained via conventional pilot training while that through passive elements are predicted via a DNN.
Based on deep unfolding, the authors of \cite{RIS_DL_Detect} has designed a model-driven DL detector for an RIS-aided spatial modulation system. It has been  shown that the proposed DL detector outperforms the traditional greedy one.

In general, the design of data-driven DNNs is a black-box strategy. 
However, the huge training overhead makes it difficult to meet the low latency and high mobility requirements of the 6G network.
Compared to data-driven DNNs, the number of trainable parameters of model-driven networks can usually be significantly reduced, at the cost of sophisticated network design based on mathematical formulations.
It is still an open problem to trade off among system performance, network complexity, and training overhead for developing interplay techniques of RIS and DL.


\subsection{Hardware implementation of RIS}
 
The performance gain of RIS is promised by adequate phase adjustment resolution of reflecting elements \cite{RIS_Continuous}.
PIN diode is the most commonly used to generate the phase shifts of RIS \cite{RIS_Practical}.
To be specific, discrete phase shifts are determined by the states of a cluster of PIN diodes with each diode switched on/off.
A PIN diode provides one-bit phase resolution \cite{RIS_Survey_Liang}.
The increment of phase resolution usually leads to an exponential growth of hardware cost, which is challenging, especially for RIS with a large number of reflecting elements.

In addition, RIS is dynamically controlled by a micro-controller and field-programmable gate array (FPGA) to generate the desired phase shifts for real-time communications.
Due to the rapidly varying wireless environments and fast moving terminals, the system performance with RIS is highly restricted by the control response time of reflecting elements.
For practical implementations, it is difficult to improve the device response bandwidth which usually results in a dramatic increment of hardware cost and power consumption.
Hence, it is an interesting problem to further consider tradeoffs between the RIS performance and the hardware cost in practical RIS deployments.

The multiplicative path loss of RIS, which acts as an inherently passive surface, inevitably leads to performance degradation especially when the RIS is located far from transceiver terminals.
Minimal active devices have been exploited to combine with this passive RIS to eliminate the effect of double path loss.
Particularly, a power amplifier has been incorporated with two RISs to form an amplifying RIS \cite{RIS_PA}. Compared to the passive RIS-assisted systems, this architecture significantly improves the energy efficiency of wireless communications.
In addition, some active devices can provide the abilities of signal processing to further enhance system performance at the cost of higher hardware and power consumption. 
The hardware design of RIS with minimal active devices is an important issue to be studied.

Despite of the aforementioned challenges, the evolution of metamaterial provides great opportunities for the development of RIS. 
One newly emerging technology is to enable the RIS to conduct mathematical calculations and artificial intelligence computations by adopting energy efficient computational metamaterials \cite{Computing_Materials}.
Based on this new hardware implementation, communication and computation can be efficiently integrated on RIS to meet the variety of requirements of the 6G network.
Great effort are expected to be taken to this promising field.



\subsection{Holographic surface}
The concept of holographic surface is arising as a new paradigm shift in RIS-assisted wireless communications.
To be specific, a holographic surface consists of an ultra large number of tiny reflecting elements densely integrated at a compact surface to yield a spatially continuous (or quasi-continuous) aperture \cite{RIS_Holo, RIS_Holo_Meta}.
In its asymptotic form, a holographic surface is considered to be equipped with an infinite number of elements separated by an infinitesimal spacing \cite{Holo_Model}.
Compared to a conventional RIS, the effective reflection area of a holographic surface is larger and thus a higher fraction of energy of incident signals can be reflected towards a desired direction.
On the other hand, the sidelobe power leakage of a holographic surface is significantly suppressed compared to that of a critically-spaced RIS.
Thanks to these advantages, a few efforts have been devoted to investigating communication systems equipped with holographic surfaces.
In \cite{RIS_Thz_Holo_GaoFF}, an effective baseband channel model has been mathematically formulated for a holographic RIS, and the beam pattern of the holographic RIS has also been analyzed in the context of a massive MIMO system.
Utilizing the derived expressions for holographic beam patterns, two beamforming schemes of holographic RIS have been proposed for channel estimation and data transmission. 
In \cite{RIS_Holo_Deng}, the joint optimization of holographic beamforming and digital precoding has been studied to maximize the sum rate of a multiuser communication system.
Moreover, the study in \cite{RIS_Holo_SE} has investigated the spectral efficiency of a multiuser holographic MIMO communication system with MRT and zero-forcing precoding.
It is revealed that the spectral efficiency increases when more elements separated by a fixed spacing are equipped at the holographic transmitter and receiver.

Despite of these studies as discussed above, there are still multiple key challenges for transceiver design of holographic surface.
On one hand, it is still extraordinarily challenging to acquire the CSI of a holographic surface composed of an ultra large number of elements.
On the other hand, the element spacing of a holographic surface is an important design parameter to be optimized.
Specifically, a narrower beamforming can be established to achieve better performance by a holographic surface with smaller element spacing.
However, when the element spacing excessively decreases, the mutual coupling effect between adjacent elements becomes overwhelmingly stronger and the signal correlations correspondingly increase, leading to inevitable performance degradation.


\section{Conclusion}
\label{Sec_conc}
In this paper, we provide an overview on RIS which is regarded as a key enabling technique for the 6G network. 
Specifically, we elaborate on two core functionalities of RIS, i.e., reflection and modulation, as well as their substantial benefits to wireless communication systems.
Then, we exposes the ability of RISs to integrate with a variety of emerging 6G applications.
To meet the prominent requirement of secure communications in the 6G network, we discuss on the contributions of RISs to enhance physical-layer security in terms of secrecy rate and SOP.
In particular, we propose a typical case study exemplifying the benefits of RISs to secure communications.
Both theoretical analysis and simulation results demonstrate the impact of RISs on the ergodic secrecy rate.
Finally, we also listed a multitude of challenges and open opportunities in order to reap the full potentials of RISs.
We hope that this overview provides a useful guide for further research in the field of RIS-assisted wireless communications.

\begin{appendix}

\section{Proof of Observation \ref{Rs_Theorem_LowB}}
\label{app_lemma_eta}

Assuming $l_B\geq l_E$ in \eqref{Rs_opt}, we have
\begin{align}
\mathbb{E}\left\{R_S\right\}&=\mathbb{E}\left\{\log_2\left(1+\frac{l_Al_B MN^2P}{\sigma_n^2}\right)-\log_2\left(1+\frac{l_Al_E MP}{\sigma_n^2}\left| \mathbf{g}_E^H\mathbf{g}_B \right|^2\right)\right\}
\nonumber\\
&\geq
\log_2\left(1+\frac{l_Al_B MN^2P}{\sigma_n^2}\right)-\log_2\left(1+\frac{l_Al_E MP}{\sigma_n^2}\mathbb{E}\left\{\left| \mathbf{g}_E^H\mathbf{g}_B \right|^2\right\}\right)
\label{Rs_bound_a}\\
&\rightarrow
\log_2\left(1+\frac{l_Al_B MN^2P}{\sigma_n^2}\right)-\log_2\left(1+\frac{l_Al_E MP}{\sigma_n^2}\eta\right),
\label{Rs_bound_b}
\end{align}
where \eqref{Rs_bound_a} uses the Jensen's inequality and \eqref{Rs_bound_b} is obtained by applying the following Lemma \ref{lemma_eta}.

\begin{lemma}
\label{lemma_eta}
For two independent channel vectors $\mathbf{g}_{B}$ and $\mathbf{g}_{E}$, it follows for large $N$ that
\begin{align}
&\mathbb{E}\left\{|\mathbf{g}_{E}^H\mathbf{g}_{B}|^2\right\}\rightarrow \eta,
\end{align}
where $\eta$ is defined as $\eta\triangleq N-\frac{2}{\pi^2}(N-1)+\frac{2N}{\pi^2}(\ln N +a)$.
\end{lemma}

\begin{proof}
According to the expressions of $\mathbf{g}_{B}$ and $\mathbf{g}_{E}$,
$\mathbb{E}\{|\mathbf{g}_{E}^H\mathbf{g}_{B}|^2\}$ is evaluated as
\begin{align}
\mathbb{E}\left\{|\mathbf{g}_{E}^H\mathbf{g}_{B}|^2\right\}
=&\mathbb{E}\left\{\!\sum_{n=0}^{N-1}e^{jn\pi\!(\!\cos\psi_E\!-\!\cos\psi_B\!)\!} \times \sum_{n=0}^{N-1}e^{-jn\pi\!(\!\cos\psi_E\!-\!\cos\psi_B\!)\!}\!\right\}\nonumber\\
=&N+\sum_{n=1}^{N-1}(N-n)
\times\left[\mathbb{E}\left\{e^{jn\pi\!(\!\cos\psi_E\!-\!\cos\psi_B\!)\!}\right\}\!+\!\mathbb{E}\left\{e^{-jn\pi\!(\!\cos\psi_E\!-\!\cos\psi_B\!)\!}\right\}\right]\nonumber\\
=&N\!+\!\sum_{n=1}^{N-1}(N\!-\!n)\!\times\!\left[\mathbb{E}\left\{e^{jn\pi\!\cos\psi_E\!\!}\right\}\mathbb{E}\left\{e^{-jn\pi\!\cos\psi_B\!\!}\right\}
+\mathbb{E}\left\{e^{-jn\pi\!\cos\psi_E\!\!}\right\}\mathbb{E}\left\{e^{jn\pi\!\cos\psi_B\!\!}\right\}\right]
\label{eta_a}\\
=&N+2\sum_{n=1}^{N-1}(N-n)J_0^2(n\pi)\label{eta_b}\\
\rightarrow &N+2\sum_{n=1}^{N-1}\frac{N-n}{n\pi^2}\label{eta_c}\\
=&N-\frac{2}{\pi^2}(N-1)+\frac{2N}{\pi^2}\sum_{n=1}^{N-1}\frac{1}{n}\nonumber\\
\rightarrow&N-\frac{2}{\pi^2}(N-1)+\frac{2N}{\pi^2}(\ln N +a),
\label{eta_d}
\end{align}
where \eqref{eta_a} uses the fact that $\psi_E$ is independent of $\psi_B$
and \eqref{eta_b} exploits the equality
\begin{align}
\mathbb{E}\left\{e^{jn\pi\cos\psi_B}\right\}=\mathbb{E}\left\{e^{-jn\pi\cos\psi_B}\right\}=\mathbb{E}\left\{e^{jn\pi\cos\psi_E}\right\}=\mathbb{E}\left\{e^{-jn\pi\cos\psi_E}\right\}=J_0(n\pi),
\end{align}
where $J_{\nu}(\cdot)$ is the ${\nu}$th Bessel function.
Take $\mathbb{E}\left\{e^{jn\pi\cos\psi_B}\right\}$ for instance.
Since $\psi_B$ follows the uniform distribution U$[0,\pi]$, we have
\begin{align}
\mathbb{E}\left\{e^{jn\pi\cos\psi_B}\right\}
&=\mathbb{E}\left\{\cos (n\pi\cos\psi_B)+j\sin (n\pi\cos\psi_B)\right\}\nonumber \\
&=\!\frac{1}{\pi}\!\int_0^{\pi}\cos (n\pi\cos\psi_B) \textrm{d}\psi_B\!+\!\frac{j}{\pi}\int_0^{\pi}\sin (n\pi\cos\psi_B) \textrm{d}\psi_B\nonumber \\
&=J_0(n\pi),
\label{J}
\end{align}
where the last step is obtained by letting $m=0, z=n\pi$, and $x=\psi_B$ in the integral equations
\begin{align}
\int_0^{\pi}\cos(z\cos x)\cos mx \textrm{d}x=\pi\cos\frac{m\pi}{2}J_m(z),
\end{align}
and
\begin{align}
\int_0^{\pi}\sin(z\cos x)\cos mx \textrm{d}x=\pi\sin\frac{m\pi}{2}J_m(z),
\end{align}
which are given in \cite[Eqs. (18), (13), pp. 425]{table}.
The step in \eqref{eta_c} exploits the asymptotical result \cite[Eq. 9.2.1]{handbook}
\begin{align}
\label{J2}
J_{\nu}(x)\rightarrow \sqrt{\frac{2}{\pi x}}\cos \left(x-\frac{\nu\pi}{2}-\frac{\pi}{4}\right),
\end{align}
for $|x|\rightarrow \infty$ with $\nu=0$.
Note that the asymptotic equality in \eqref{J2} is tight even for small $x$.
Finally, \eqref{eta_d} follows by large $N$ and the definition of the Euler's constant \cite{handbook} as
\begin{align}
\label{Euler}
a\triangleq\lim_{n\rightarrow\infty} \left[\sum_{k=1}^{k=n-1}\frac{1}{k}-\ln n\right].
\end{align}
\end{proof}

\end{appendix}

\end{document}